\documentclass[preprint]{elsarticle}

\usepackage[english]{babel}
\usepackage{enumerate}
\usepackage{color}
\usepackage{subfigure}
\usepackage{dsfont}
\usepackage{graphicx}
\usepackage{epstopdf}
\usepackage{amsmath}
\usepackage{mathtools}
\usepackage{amsthm}
\usepackage{amstext}
\usepackage{amssymb}
\usepackage{mathrsfs}
\usepackage{tikz}
\usetikzlibrary{arrows,positioning}

\def\R{\mathbb{R}}

\def\rP{\mathbb{P}}

\def\L{\mathop{\rm Law}}

\def\P{{\mathcal P}}

\def\bm{{\boldsymbol m}}

\def\bmu{{\boldsymbol \mu}}
\def\bnu{{\boldsymbol \nu}}

\def\sX{{\mathsf X}}
\def\sY{{\mathsf Y}}
\def\sA{{\mathsf A}}

\def\sE{{\mathsf E}}

\def\sW{{\mathsf W}}

\def\sV{{\mathsf V}}

\theoremstyle{remark}

\newtheorem{definition}{Definition}
\newtheorem{theorem}{Theorem}

\newtheorem{proposition}{Proposition}
\newtheorem{lemma}{Lemma}
\theoremstyle{remark}
\newtheorem{remark}{Remark}
\newtheorem{assumption}{Assumption}

\def\Tiny{\fontsize{4pt}{4pt}\selectfont}
\newcommand*{\eqdef}{\ensuremath{\overset{\mathclap{\text{\Tiny def}}}{=}}}

\allowdisplaybreaks

\begin{document}
\begin{frontmatter}
\sloppy
\title{Large Deviations Principle for Discrete-time Mean-field Games}
\author[NS]{Naci Saldi}
\address[NS]{\"{O}zye\u{g}in University, \c{C}ekmek\"{o}y, \.{I}stanbul, Turkey, Email:naci.saldi@ozyegin.edu.tr}

\begin{abstract}
In this paper, we establish a large deviations principle (LDP) for interacting particle systems that arise from state and action dynamics of discrete-time mean-field games under the equilibrium policy of the infinite-population limit. The LDP is proved under weak Feller continuity of state and action dynamics.
The proof is based on transferring LDP for empirical measures of initial states and noise variables under setwise topology to the original game model via contraction principle, which was first suggested by Delarue, Lacker, and Ramanan to establish LDP for continuous-time mean-field games under common noise. We also compare our work with LDP results established in prior literature for interacting particle systems, which are in a sense uncontrolled versions of mean-field games.  
\end{abstract}


\end{frontmatter}

\section{Introduction}\label{sec0}

The purpose of this paper is to establish a large deviations principle (LDP) for interacting particle systems that arise from state and action dynamics of discrete-time mean-field games under the infinite-population equilibrium policies. Namely, by assuming the existence of an equilibrium in the infinite-population, our goal is to establish exponential decline of the large deviation probabilities of empirical state-action measures from limiting distribution when each agent applies infinite-population equilibrium policy in the finite agent setting.

Large deviations principle for mean-field games has been studied recently in \cite{LaRa19,DeLaRa20}. In \cite{LaRa19}, authors consider a static mean-field game with a centralized information structure (i.e., each agent has access to the entire type vector). They establish that sets of empirical measures of type and action vectors under Nash equilibria of the finite agent game, where in this case it is indeed feasible to attain Nash equilibria due to the centralized information structure, satisfy LDP. In \cite{DeLaRa20}, authors study a continuous-time mean-field game under a common noise with again a centralized information structure (i.e., each agent has access to the entire state vector). In this setting, owing to the centralized information structure, it is possible to characterize the unique Nash equilibria of the finite-agent games by solving the so-called Nash systems. Then, they demonstrate that under the above-mentioned unique Nash equilibria, empirical measures of states satisfy LDP. The main challenge in this work is the common noise that affects dynamics of the every agent. 

It is important to note that our LDP result is intrinsically different than LDP results in \cite{LaRa19,DeLaRa20}, since we establish LDP for empirical measures of state-action pairs under the infinite-population equilibrium policy, whereas in \cite{LaRa19,DeLaRa20}, LDP results are proved under finite-population equilibrium policies. Therefore, in our case, the  mean-field game model reduces to a weakly interacting particle system. The first reason for this distinction is that under decentralized information structure, we need quite restrictive assumptions to establish the existence of a Nash equilibrium in the finite-population setting for discrete-time mean-field games. Secondly, mean-field game theory has been developed to approximate Nash equilibria of finite-population games, which are hard to compute, via equilibrium in the infinite-population limit. The approximation result depends highly on the asymptotic convergence of the empirical measures of state-action pairs under the infinite-population equilibrium policy to the limiting distribution. Therefore, it is also interesting to study the large deviations behavior of the same empirical measures of state-action pairs.

In continuous-time, one of the classical references for LDP of interacting particle systems is \cite{DaGa87}, which establishes a LDP for empirical measures of states using discretization and projective limit arguments. Another work toward that direction is \cite{BuDuFi12}, where a LDP principle for empirical measures of states is demonstrated via weak convergence method. We refer the reader to \cite{BuDuFi12,DeLaRa20} for a more comprehensive literature review on continuous-time setup. There are also some works on LDP for discrete-time weakly interacting particle systems \cite{MoZa03,DaMo05,MoGu98}, which are intimately related to our model. However, in those works, LDP results are established under more restrictive conditions on system components than ours. Moreover, in \cite{DaMo05,MoGu98}, dynamics of interacting particle systems are slightly different than the system dynamics considered in this paper, and so,  corresponding LDP results are incompatible. We refer the reader to Section~\ref{sec2} and Remark~\ref{remark1} to see the precise distinction between our work and those works. 


\noindent\textbf{Notation.} 
For a metric space $\sE$, we let $\P(\sE)$, $C_b(\sE)$, and $B(\sE)$ denote the set of all Borel probability measures, the set of all bounded and continuous real-valued functions, and the set of all bounded and Borel measurable real-valued functions on $\sE$, respectively. A sequence or a net $\{\mu_{\alpha}\}$ of probability measures on $\sE$ is said to converge weakly (setwise) to a probability measure $\mu$ if $\int_{\sE} g(e) \mu_{\alpha}(de)\rightarrow\int_{\sE} g(e) \mu(de)$ for all $g \in C_b(\sE)$ (for all $g \in B(\sE)$). The set of probability measures $\P(\sE)$ is endowed with the Borel $\sigma$-algebra induced by the weak convergence topology. For any $0<T\leq\infty$, let $\sE^T \eqdef \prod_{t=0}^T \sE$. Given $\Lambda \in \P(\sE^T)$, let $\Lambda^{z(t)} \eqdef \Lambda(\,\cdot\,|z(t))$ denote the conditional distribution of $\{z(k)\}_{k\neq t}^T$ given $z(t)$ under $\Lambda$. For any $\nu \in \P(\sE_1\times\sE_2)$, where $\sE_1$ and $\sE_2$ are metric spaces, we let $\nu_{\sE_i}$ denote the marginal of $\nu$ on $\sE_i$, $i=1,2$. 
A sequence or a net $\{\mu_{\alpha}\}$ of probability measures on $\sE_1\times\sE_2$ is said to converge in $ws$-topology to a probability measure $\mu$ if $\int_{\sE_1\times\sE_2} g(e_1,e_2) \mu_{\alpha}(de_1,de_2)\rightarrow\int_{\sE_1\times\sE_2} g(e_1,e_2) \mu(de_1,de_2)$ for all $g \in B(\sE_1\times\sE_2)$ continuous on $\sE_1$.
The notation $v\sim \nu$ means that the random element $v$ has distribution $\nu$. Unless  specified otherwise, the term "measurable" will refer to Borel measurability.

\section{Game Model}\label{sec1}

Although the main result is a LDP for discrete-time interacting particle system, we still need to briefly describe the mean-field game model as the interacting particle system arises from its state-action dynamics under the infinite-population equilibrium policy.

In the game model, we have $N$-agents with an identical state space $\sX$ and an identical action space $\sA$, where $\sX$ and $\sA$ are Polish spaces. For every $t \in \{0,1,2,\ldots\}$ and every $i \in \{1,2,\ldots,N\}$, the random elements $x^N_i(t) \in \sX$ and $a^N_i(t) \in \sA$ denote the state and the action of Agent~$i$ at time $t$. The state $x_i^N(t)$ of Agent~$i$ evolves as follows:
\begin{align}
x_i^N(0) \sim \mu(0), \,\,\,\, x_i^N(t+1) = F_t\left(x_i^N(t),a_i^N(t),e^N(t),w_i^N(t)\right), \nonumber
\end{align}
where $F_t: \sX \times \sA \times \P(\sX) \times \sW \rightarrow \sX$ is a measurable function. Here,  
$
e^{N}(t)(\,\cdot\,) \eqdef \frac{1}{N} \sum_{i=1}^N \delta_{x_i^N(t)}(\,\cdot\,) \in \P(\sX) 
$
is the empirical measure of the states at time $t$, where $\delta_x\in\P(\sX)$ is the Dirac measure at $x$, and $w_i^N(t) \in \sW$ is the noise with distribution $\xi_t$. The initial states $x^N_i(0)$ are independent and identically distributed (i.i.d.) according to $\mu(0)$, and noise variables $\{w_i^N(t)\}_{t\geq0}$ are independent in time $t$ and in position $i$, and also independent of initial states.

A (Markov) \emph{policy} $\pi^i = \{\pi_t^i\}_{t\geq0} \in \Pi_i$ for Agent~$i$ is a sequence of stochastic kernels on $\sA$ given $\sX$.
Therefore, given the current state configuration $\left(x_1^N(t),\ldots,x_N^N(t)\right)$, the actions are generated as follows:
$
\left(a_1^N(t),\ldots,a_N^N(t)\right) \sim \bigotimes^N_{i=1} \pi^i_t\big(\,\cdot\,\big|x^N_i(t)\big). 
$
For each Agent~$i$, the initial distribution $\mu(0)$ and the $N$-tuple of policies $(\pi^1,\ldots,\pi^N)$ induces a cost function $W_i^{N}(\pi^1,\ldots,\pi^N)$  (e.g., discounted cost, average cost, finite-horizon cost). An $N$-tuple of policies ${\boldsymbol \pi}^{(N*)}= (\pi^{1*},\ldots,\pi^{N*})$ is a \emph{Nash equilibrium} if
$
W_i^{N}({\boldsymbol \pi}^{(N*)}) = \inf_{\pi^i \in \Pi_i} W_i^{N}({\boldsymbol \pi}^{(N*)}_{-i},\pi^i) \nonumber
$
for each $i=1,\ldots,N$, where ${\boldsymbol \pi}^{(N*)}_{-i} \eqdef (\pi^{j*})_{j\neq i}$.

For this game model, it is, in general, infeasible to obtain a Nash equilibrium due to the decentralized information structure (i.e., each agent has only access to its own state) and the large number of coupled agents. However, if the number of agents is large enough, a canonical method to deal with this challenge is to introduce the infinite-population limit $N\rightarrow\infty$ of the game to obtain an approximate Nash equilibrium. 

In the infinite population game model, if we suppose that the empirical distribution $e^N(t)$ converges to the deterministic  measure $\mu(t)$ for each $t\geq0$, the state $x(t)$ of a generic agent evolves as follows:
\begin{align}
x(0) \sim \mu(0), \,\,\,\, x(t+1) = F_t\left(x(t),a(t),\mu(t),w(t)\right). \nonumber 
\end{align}
Let $\bmu \eqdef (\mu(t))_{t\geq0}$, which describes the collective behavior of all agents in the infinite population limit. In this model, a (Markov) policy $\pi = \{\pi_t\}_{t\geq0} \in \Pi$ of a generic agent is again a sequence of stochastic kernels on $\sA$ given $\sX$. A policy $\pi^{*} \in \Pi$ is optimal for $\bmu$ if
$
W_{\bmu}(\pi^{*}) = \inf_{\pi \in \Pi} W_{\bmu}(\pi), 
$
where $W_{\bmu}$ is obtained by replacing $e^N(t)$ with $\mu(t)$ in $W_i^{N}$. The equilibrium notion for the infinite-population game is defined as follows. A pair $(\pi^{\bm},\bmu^{\bm})$ is a \emph{mean-field equilibrium} if $\pi^{\bm}$ is optimal for $\bmu^{\bm}$ and $\bmu^{\bm}$ is the collection of state distributions under the policy $\pi^{\bm}$.

In the literature, the existence of mean-field equilibrium has been established under mild assumptions (see \cite{SaBaRaSIAM}), which will not be stated here. Instead, we suppose that there exists a mean-field equilibrium $(\pi^{\bm},\bmu^{\bm})$. Note that one can always write the evolution of the action process $a(t)$ under $\pi^{\bm}$ as a noise-driven dynamical system
$$
a(t) = G_t(x(t),v(t)),
$$
where $G_t:\sX \times \sV \rightarrow \sA$ is measurable and $\{v(t)\}_{t\geq0}$ is an independent noise process on some Polish space $\sV$ with $v(t) \sim \alpha_t$ for each $t\geq0$. To establish the large deviations principle, we impose the following assumption.

\begin{assumption}\label{assumption1}
We suppose that for each $t\geq0$, the mapping $F_t(\,\cdot\,,w)$ is continuous in $(x,a,\mu) \in \sX \times \sA \times \P(\sX)$ for any $w \in \sW$, where $\P(\sX)$ is endowed with weak topology, and the mapping $G_t(\,\cdot\,,v)$ is continuous in $x \in \sX$ for any $v \in \sV$. 
\end{assumption}

Note that given the noise variable $w$, continuity of $F_t$ with respect to $(x,a,\mu)$ corresponds to the weak Feller continuity of the transition probability $p_t^{\mu}(\,\cdot\,|x,a)$, where the transition probability is defined as: 
$$
p_t^{\mu}(\,\cdot\,|x,a) \eqdef \int_{\sW} \delta_{F(x,a,\mu,w)}(\,\cdot\,) \, \xi_t(dw). 
$$
This is a quite weak assumption and is satisfied by most of the systems in real-life. However, the continuity assumption on $G_t$ is a bit restrictive. For discounted-cost and finite-horizon cost, one way to establish the continuity of $G_t$ is to assume that $\sX = \R^d$ and $\sA \subset \R^m$ is convex. Moreover, suppose that $p^{\mu}_t(dx'|x,a) = \varrho^{\mu}_t(x'|x,a) \, m(dx')$, where $m$ denotes the Lebesgue measure and $\varrho^{\mu}_t$ is the corresponding density function. Then, we assume that both $\varrho_t$ and $c_t$ are strictly convex in $a$ (see \cite[Remark 7]{SaBaRaSIAM}).

\subsection{Main Result}\label{sub2sec1}

Note that the mean-field equilibrium policy $\pi^{\bm}$, if it is used by all agents, constitutes an approximate Nash equilibrium for the $N$-agent game model if $N$ is sufficiently large. This result depends highly on the following law of large numbers (LLN) principle. To state the LLN,  for each $t\geq0$, let $b^{\bm}(t) \eqdef \mu^{\bm}(t)\otimes\pi^{\bm}_t$ and for the $N$-agent game model, we define the empirical measure of state-action pairs under policy $(\pi^{\bm},\ldots,\pi^{\bm})$ as follows:
$$
b^N(t) \eqdef \frac{1}{N} \sum_{i=1}^N \delta_{\left(x_i^N(t),a_i^N(t)\right)}. 
$$ 

\begin{theorem}\label{theorem1}
Under Assumption~\ref{assumption1}, for each $t\geq0$, the empirical measure $b^N(t)$ converges in distribution to the deterministic measure $b^{\bm}(t)$ as $N \rightarrow \infty$. 
\end{theorem}

\begin{proof}
The proof can be done as in the proof of \cite[Proposition 4.4]{SaBaRaSIAM}. 
\end{proof}

In this paper, we are interested in deviations of empirical measures $(b^N(t))_{t\geq0}$ from their limits $(b^{\bm}(t))_{t\geq0}$. To this end, we establish the below large deviations principle, which is the main result of this paper. Before, we state the main result, let us recall the definition of LDP and contraction principle. Note that a function $I:\sE\rightarrow[0,\infty]$ on a Hausdorff space $\sE$ is called a \emph{rate function} if $\{z \in \sE: I(z) \leq M\}$ is  compact for each $M<\infty$.

\begin{definition}
A process $\{Z_N\}_{N\geq1}$ satisfies LDP on a Hausdorff space $\sE$ with a rate function $I:\sE\rightarrow[0,\infty]$ if for each closed subset $F \subset \sE$ and for each open subset $G \subset \sE$, we have
\begin{align}
\limsup_{N\rightarrow\infty} \frac{1}{N} \log \rP\{Z_N \in F\} &\leq -\inf_{z\in F} I(z) \,\, \text{(LDP upper bound)}\nonumber \\ 
\liminf_{N\rightarrow\infty} \frac{1}{N} \log \rP\{Z_N \in G\} &\geq -\inf_{z\in G} I(z) \,\, \text{(LDP lower bound)}.\nonumber
\end{align}
\end{definition}

\begin{theorem}[contraction principle]\label{contraction}
Let $\sX$ and $\sY$ be Hausdorff spaces, $I$ be a rate function on $\sX$, and $f$ is a continuous mapping from $\sX$ to $\sY$. Then, we have the following conclusions.
\begin{itemize}
\item[(a)] For each $y \in \sY$, $J(y) \eqdef \inf \{I(x): f(x)=y\}$ is a rate function on $\sY$.
\item[(b)] If a process $\{Z_N\}_{N\geq1}$ satisfies LDP on $\sX$ with a rate function $I:\sE\rightarrow[0,\infty]$, then the process  $\{f(Z_N)\}_{N\geq1}$ satisfies LDP on  $\sY$ with a rate function $J:\sY\rightarrow[0,\infty]$.
\end{itemize} 
\end{theorem}

Now we can state the main result of this paper. Here, we suppose that $\P(\sX\times\sA)^{\infty}$ is endowed with the product topology, where $\P(\sX\times\sA)$ has weak convergence topology.  

\begin{theorem}\label{theorem2}
Under Assumption~\ref{assumption1}, empirical measures $\{(b^N(t))_{t\geq0}\}_{N\geq1}$ satisfy large deviations principle (LDP) on $\P(\sX\times\sA)^{\infty}$ with the following rate function:
\begin{align}
J(\bnu) \eqdef \inf\left\{ R\left(\Lambda \, \big|\,\Lambda_{\nu(0)}^{\bnu}\right): \Lambda \in \P(\sX^{\infty}\times\sA^{\infty}), \,\, \Lambda_t = \nu_t \, \forall \, t\geq0 \right\}, \nonumber
\end{align}
where $\bnu = (\nu_t)_{t\geq0}$, $R(\,\cdot\,|\,\cdot\,)$ is the relative entropy, 
\begin{align}
&\Lambda_{\nu(0)}^{\bnu}\left(dx(0),da(0),dx(1),da(1),\ldots\right) \eqdef \nu(0)(dx(0),da(0)) \nonumber \\
&\phantom{xxxxxxxxx}\bigotimes_{t=0}^{\infty}\bigg( p_t^{\nu_{t,\sX}}(dx(t+1)|x(t),a(t))\otimes\pi_{t+1}^{\bm}(da(t+1)|x(t+1))\bigg), \nonumber \\
&\nu(0)(dx(0),da(0)) \eqdef \mu^{\bm}(0)(dx(0))\otimes\pi_0^{\bm}(da(0)|x(0)),\nonumber 
\end{align}
 and $\Lambda_k \in \P(\sX\times\sA)$ is the $k^{th}$-marginal of $\Lambda$ on $\sX\times\sA$.
\end{theorem} 

Before we compare our result with the prior literature, let us analyze the rate function when the time horizon is $T=2$ instead of infinity. In this case, the rate function becomes 
$$
J(\nu_0,\nu_1,\nu_2) = \inf\left\{ R\left(\Lambda \, \big|\,\Lambda_{\nu(0)}^{\bnu}\right): \Lambda \in \P(\sX^{3}\times\sA^{3}), \,\, \Lambda_t = \nu_t \, \forall \, t=0,1,2 \right\}, \nonumber
$$
where $\Lambda_{\nu(0)}^{\bnu}$ is the joint distribution of $(x(0),a(0),x(1),a(1),x(2),a(2))$ in the infinite-population limit of the game under the measure flow $(\nu_{0,\sX},\nu_{1,\sX},\nu_{2,\sX})$, the policy $(\pi_0^{\bm},\pi_1^{\bm},\pi_2^{\bm})$, and the initial distribution $\mu^{\bm}(0)$. One can also prove that (see Proposition~\ref{proposition1}) 
$$
J(\nu_0,\nu_1,\nu_2) \geq R(\nu_0|\nu(0)) + R(\nu_1|\Gamma_0(\nu_0)) + R(\nu_2|\Gamma_1(\nu_1)),  
$$
where $\Gamma_t(\nu)(dx',da') = \int_{\sX\times\sA} \pi_{t+1}^{\bm}(da'|x') \, p_t^{\nu_{\sX}}(dx'|x,a)\, \nu(dx,da)$. Therefore, for $t=0,1$, if $\nu_{t+1}$ is quite different than the distribution of $(x(t+1),a(t+1))$ under the policy $\pi_{t+1}^{\bm}$ and transition probability $p_t^{\nu_{t,\sX}}$ when $(x(t),a(t)) \sim \nu_t$, then the rate function becomes large. Hence, convergence of empirical measures of state-action pairs  to such limiting distributions is highly unlikely.  

\section{Comparison with Prior Literature}\label{sec2}

In \cite{MoGu98,DaMo05}, authors analyze the following interacting particle system. Let $\sE$ be a Polish space. For each $t\geq0$, consider the nonlinear mapping 
$\Gamma_t: \P(\sE) \rightarrow \P(\sE)$ given by
$$
\Gamma_t(\nu(t))(\,\cdot\,) = \int_{\sE} \kappa_t^{\nu(t)}(\,\cdot\,|z) \, \nu(t)(dz) \eqdef \nu(t)\otimes\kappa_t^{\nu(t)},  
$$
for some $\kappa_t^{\nu}(\,\cdot\,|z): \sE \times \P(\sE) \rightarrow \P(\sE)$.
Let $\{\epsilon^N(t)\}_{t\geq0}$ be a Markov chain with the state space $\sE^N \eqdef \prod_{i=1}^N \sE$, the initial distribution $\bigotimes_{i=1}^N\nu_0$, and the transition probability 
\begin{align}
&\rP\left\{ \epsilon^N(t+1) \in dz_1^N(t+1),\ldots,dz_N^N(t+1) \,\big|\, \epsilon^N(t)\right\} \nonumber \\
&\phantom{xxxxxxxxxxxxxxxxx}= \bigotimes_{i=1}^N \frac{1}{N} \sum_{j=1}^N \kappa_t^{b^N(t)}(dz_i^N(t+1)\,|\, z_j^N(t)), \label{eq2}
\end{align}
where $b^N(t) \in \P(\sE)$ is the empirical distribution of the vector $\epsilon^N(t) \in \sE^N$. Note that the state-action dynamics in our game model is \emph{almost} the same with (\ref{eq2}). It is almost the same because, in our case, in place of (\ref{eq2}), we have the following dynamics:
\begin{align}\label{eq3}
&\rP\left\{ \epsilon^N(t+1) \in dz_1^N(t+1),\ldots,dz_N^N(t+1) \,\big|\, \epsilon^N(t)\right\} \nonumber \\
&\phantom{xxxxxxxxxxxxxxxxx}= \bigotimes_{i=1}^N \kappa_t^{b^N(t)}(dz_i^N(t+1)\,|\, z_i^N(t)), 
\end{align}
where $z_i^N(t) = (x_i^N(t),a_i^N(t)) \in \sX\times\sA \eqdef \sE$ and 
$$\kappa_t^{\nu}(dz'|z) = p_t^{\nu_{\sX}}(dx'|x,a)\otimes  \pi_{t+1}^{\bm}(da'|x'),$$ 
$z' = (x',a')$, and $z = (x,a)$. Because of this distinction, the analysis of our system is different than the one in (\ref{eq2}).

Let us consider the model in (\ref{eq2}) with $\kappa_t^{\nu}(dz'|z) = p_t^{\nu_{\sX}}(dx'|x,a) \otimes \pi_{t+1}^{\bm}(da'|x') $. In \cite{DaMo05}, under minorization and absolute continuity conditions on the nonlinear mapping $\Gamma_t$ (which are more restrictive than Assumption~\ref{assumption1}), it has been established that for any $T\geq0$, the empirical process $\{(b^N(t))_{t=0}^T\}_{N\geq1}$ satisfies LDP on $\P(\sE)^T$ with the rate function 
$$
V_T(\gamma_0,\ldots,\gamma_T) = R(\gamma_0|\nu(0)) + \sum_{t=1}^T R(\gamma_t|\Gamma_{t-1}(\gamma_{t-1})).
$$ 
Note that under our model (\ref{eq3}), by Theorem~\ref{theorem2} and Theorem~\ref{contraction} (contraction principle), the same empirical process $\{(b^N(t))_{t=0}^T\}_{N\geq1}$ satisfies the LDP on $\P(\sE)^T$ with the rate function 
$$
J_T(\gamma_0,\ldots,\gamma_T) = \inf\left\{R(\Lambda|\Lambda_{\nu(0)}^{\bf \gamma}): \Lambda \in \P(\sE^T), \, \Lambda_t = \gamma_t \, \forall t=0,\ldots,T\right\},
$$     
where 
\begin{align}
&\Lambda_{\nu(0)}^{\bf \gamma}(dx(0),da(0),\ldots,dx(T),da(T))\eqdef \nu(0)(dx(0),da(0))  \nonumber \\
&\otimes \kappa_0^{\gamma_0}(d x(1),d a(1)|x(0),a(0))\otimes\ldots\otimes\kappa_{T-1}^{\gamma_{T-1}}(dx(T),da(T)|x(T-1),a(T-1)) \nonumber 
\end{align}
and $\Lambda_k \in \P(\sE)$ is the $k^{th}$-marginal of $\Lambda$ on $\sE$.
In \cite[p. 181]{DaMo05}, although $J_T$ could not be characterized, it has been suggested that $J_T$ should be greater than $V_T$. This is indeed the case as will be shown below. 
Before we state the result, let us introduce the following notation. For each $t\geq0$, define 
$$
\theta_{t,t+1}(dz(t+1),dz(t)) \eqdef \kappa_t^{\gamma_t}(dz(t+1)|z(t)) \, \gamma_t(dz(t)). 
$$
If we disintegrate $\theta_{t,t+1}(dz(t+1),dz(t))$ in the reverse order, we obtain the following stochastic kernel
$$
\bar{\kappa}_t^{\gamma_t}(dz(t)|z(t+1)) \eqdef \theta_{t,t+1}(dz(t)|z(t+1)),
$$
which obviously depends on $\gamma_t$. 

\begin{proposition}\label{proposition1}
For any $T\geq0$, we have $J_T \geq V_T$. Indeed, we have the following relation:
\small
$$
J_T(\gamma_0,\ldots,\gamma_T) = V_T(\gamma_0,\ldots,\gamma_T) + \inf\left\{R(\Lambda|\bar{\Lambda}^{\bf \gamma}): \Lambda \in \P(\sE^T), \, \Lambda_t = \gamma_t \, \forall t=0,\ldots,T\right\}, \nonumber 
$$
\normalsize
where 
\begin{align}
&\bar{\Lambda}^{\bf \gamma}(dx(T),da(T),\ldots,dx(0),da(0)) \eqdef \gamma(T)(dx(T),da(T)) \nonumber \\
&\otimes \bar{\kappa}_{T-1}^{\gamma(T-1)}(dx(T-1),a(T-1)|x(T),a(T))\otimes\ldots\otimes\bar{\kappa}_0^{\gamma_0}(dx(0),da(0)|dx(1),da(1)).\nonumber 
\end{align} 
\end{proposition}

\begin{proof}
We give the proof for $T=3$ and leave the general case as an exercise to the reader. Let $\Lambda \in \P(\sE^3)$ with $\Lambda_t = \gamma_t$ for all $t=0,1,2,3$. By repeatedly applying the chain rule \cite[Theorem C.3.1]{DuEl97} for relative entropy, we obtain 
\begin{align}
&R(\Lambda|\Lambda_{\nu(0)}^{\bf \gamma}) \nonumber \\
&= R(\gamma_0|\nu(0)) + \int_{\sE} R(\Lambda^{z(0)}|\kappa_0^{\gamma_0}(\,\cdot\,|z(0))\otimes\kappa_1^{\gamma_1}\otimes\kappa_{2}^{\gamma_{2}}) \, \gamma_0(dz(0)) \nonumber \\
&= R(\gamma_0|\nu(0)) + R(\Lambda|\gamma_0 \otimes \kappa_0^{\gamma_0}\otimes\kappa_1^{\gamma_1}\otimes\kappa_{2}^{\gamma_{2}}) \nonumber \\
&= R(\gamma_0|\nu(0)) + R(\gamma_1|\Gamma_0(\gamma_0)) \nonumber \\
&\phantom{xxxxx}+ \int_{\sE} R(\Lambda^{z(1)}|\kappa_1^{\gamma_1}(\,\cdot\,|z(1))\otimes\bar{\kappa}_0^{\gamma_0}(\,\cdot\,|z(1))\otimes\kappa_2^{\gamma_2}) \, \gamma_1(dz(1)) \nonumber \\
&= R(\gamma_0|\nu(0)) + R(\gamma_1|\Gamma_0(\gamma_0)) + R(\Lambda|\gamma_1\otimes(\kappa_1^{\gamma_1}\otimes\bar{\kappa}_0^{\gamma_0})\otimes\kappa_2^{\gamma_2}) \nonumber \\
&= R(\gamma_0|\nu(0)) + R(\gamma_1|\Gamma_0(\gamma_0)) + R(\gamma_2|\Gamma_1(\gamma_1)) \nonumber \\
&\phantom{xxxxx}+ \int_{\sE} R(\Lambda^{z(2)}|\bar{\kappa}_1^{\gamma_1}(\,\cdot\,|z(2))\otimes\kappa_2^{\gamma_2}(\,\cdot\,|z(2))\otimes\bar{\kappa}_0^{\gamma_0}) \, \gamma_2(dz(2)) \nonumber \\
&= R(\gamma_0|\nu(0)) + R(\gamma_1|\Gamma_0(\gamma_0)) + R(\gamma_2|\Gamma_1(\gamma_1))+ R(\Lambda|\gamma_2\otimes(\bar{\kappa}_1^{\gamma_1}\otimes\kappa_2^{\gamma_2})\otimes\bar{\kappa}_0^{\gamma_0}) \nonumber \\
&= R(\gamma_0|\nu(0)) + R(\gamma_1|\Gamma_0(\gamma_0)) + R(\gamma_2|\Gamma_1(\gamma_1)) + R(\gamma_3|\Gamma_2(\gamma_2)) \nonumber \\
&\phantom{xxxxx}+ \int_{\sE} R(\Lambda^{z(3)}|\bar{\kappa}_2^{\gamma_2}(\,\cdot\,|z(3))\otimes\bar{\kappa}_1^{\gamma_1}\otimes\bar{\kappa}_0^{\gamma_0}) \, \gamma_3(dz(3)) \nonumber \\
&= R(\gamma_0|\nu(0)) + R(\gamma_1|\Gamma_0(\gamma_0)) + R(\gamma_2|\Gamma_1(\gamma_1))+ R(\gamma_3|\Gamma_2(\gamma_2)) \nonumber \\
&\phantom{xxxxx}+R(\Lambda|\gamma_3\otimes\bar{\kappa}_2^{\gamma_2}\otimes\bar{\kappa}_1^{\gamma_1}\otimes\bar{\kappa}_0^{\gamma_0}). \nonumber 
\end{align}
By taking the infimum of the above identity with respect to $\Lambda$, we complete the proof for $T=3$.
\end{proof}

Consider again the model in (\ref{eq2}) with $\kappa_t^{\nu}(dz'|z) = p_t^{\nu_{\sX}}(dx'|x,a) \otimes \pi_{t+1}^{\bm}(da'|x') $. In \cite{MoGu98}, under continuity of $\{\Gamma_t\}_{t\geq0}$ and a condition that ensures the exponential tightness of $\{b^N(t)\}_{N\geq1}$ for all $t\geq0$ (which are again more restrictive than Assumption~\ref{assumption1}), it has been proved that for all $t\geq0$, the empirical measures $\{b^N(t)\}_{N\geq1}$ satisfy the LDP on $\P(\sE)$ with the rate function 
$$
V_t(\gamma) = \sup_{g \in C_b(\sE)} \left\{\int_{\sE} g(z) \, \gamma(dz) + \inf_{\nu \in \P(\sE)} \left( V_{t-1}(\nu) - \log \int_{\sE} e^{g(z)} \, \Gamma_{t-1}(\nu)(dz) \right) \right\}.
$$
Here, $V_0(\gamma) = R(\gamma|\nu(0))$. If we consider our model (\ref{eq3}) instead of (\ref{eq2}), then by Theorem~\ref{theorem2} and Theorem~\ref{contraction} (contraction principle), for each $t\geq0$, the same empirical measures $\{b^N(t)\}_{N\geq1}$ satisfy the LDP on $\P(\sE)$ with the rate function 
$$
J_t(\gamma) = \inf\left\{R(\Lambda|\Lambda_{\nu(0)}^{\Lambda}): \Lambda \in \P(\sE^t), \, \Lambda_t = \gamma\right\},
$$     
where 
\begin{align}
&\Lambda_{\nu(0)}^{\Lambda}(dx(0),da(0),\ldots,dx(t),da(t))\eqdef \nu(0)(dx(0),da(0))  \nonumber \\
&\phantom{xx}\otimes \kappa_0^{\Lambda_0}(d x(1),d a(1)|x(0),a(0))\otimes\ldots\otimes\kappa_{t-1}^{\Lambda_{t-1}}(dx(t),da(t)|x(t-1),a(t-1)) \nonumber 
\end{align}
In this case, $J_t$ is again greater than $V_t$ for each $t\geq0$ as will be shown below. 

\begin{proposition}\label{proposition2}
For any $t\geq0$, we have $J_t \geq V_t$.
\end{proposition}
\begin{proof}
For $t=0$, $J_0=V_0$. Hence the claim holds for $t=0$. Suppose that it is true for all $t\leq k$ and consider $k+1$. Let $\Lambda \in \P(\sE^{k+1})$ with $\Lambda_{k+1} = \gamma$. By chain rule \cite[Theorem C.3.1]{DuEl97} for relative entropy, we can write 
\begin{align}
R(\Lambda|\Lambda_{\nu(0)}^{\Lambda}) &= R(\Lambda_0^k|\Lambda_{\nu(0),0}^{\Lambda,k}) + R(\Lambda|\Lambda_0^k\otimes\kappa_k^{\Lambda_k}) \nonumber \\
&= R(\Lambda_0^k|\Lambda_{\nu(0),0}^{\Lambda,k}) + R(\gamma|\Gamma_k(\Lambda_k)) + R(\Lambda|\gamma\otimes\bar{\kappa}_k^{\Lambda_0^k}), \nonumber 
\end{align}
where $\Lambda_0^k$ and $\Lambda_{\nu(0),0}^{\Lambda,k}$ are marginals of $\Lambda_0$ and $\Lambda_{\nu(0)}^{\Lambda}$ on $\sE^k$, and 
$$
(\Lambda_0^k\otimes\kappa_k^{\Lambda_k})(d(z(0),\ldots,z(k))\,|\,z(k+1)) \eqdef \bar{\kappa}_k^{\Lambda_0^k}(d(z(0),\ldots,z(k))\,|\,z(k+1)).
$$ 
By Donsker-Varadhan variational formula \cite[Lemma 1.4.3]{DuEl97}, we then have
\begin{align}
&J_{k+1}(\gamma) = \inf\left\{R(\Lambda|\Lambda_{\nu(0)}^{\Lambda}): \Lambda \in \P(\sE^{k+1}), \, \Lambda_{k+1} = \gamma\right\} \nonumber \\
&= \inf\left\{R(\Lambda_0^k|\Lambda_{\nu(0),0}^{\Lambda,k})+R(\gamma|\Gamma_k(\Lambda_k))+R(\Lambda|\gamma\otimes\bar{\kappa}_k^{\Lambda_0^k}): \Lambda \in \P(\sE^{k+1}), \, \Lambda_{k+1} = \gamma\right\} \nonumber \\
&\geq  \inf\left\{R(\Lambda_0^k|\Lambda_{\nu(0),0}^{\Lambda,k})+R(\gamma|\Gamma_k(\Lambda_k)): \Lambda \in \P(\sE^{k+1}), \, \Lambda_{k+1} = \gamma\right\} \nonumber \\
&= \inf_{\substack{\Lambda \in \P(\sE^{k+1}) \\ \Lambda_{k+1} = \gamma}}\left\{R(\Lambda_0^k|\Lambda_{\nu(0),0}^{\Lambda,k})+\sup_{g \in C_b(\sE)}\left(\int_{\sE} g(z) \, \gamma(dz)-\log \int_{\sE} e^{g(z)} \, \Gamma_k(\Lambda_k)(dz)\right)\right\} \nonumber \\
&\geq \sup_{g \in C_b(\sE)} \left\{\int_{\sE} g(z) \, \gamma(dz)+ \inf_{\substack{\Lambda \in \P(\sE^{k+1}) \\ \Lambda_{k+1} = \gamma}}\left(R(\Lambda_0^k|\Lambda_{\nu(0),0}^{\Lambda,k})-\log \int_{\sE} e^{g(z)} \, \Gamma_k(\Lambda_k)(dz)\right)\right\} \nonumber \\
&= \sup_{g \in C_b(\sE)} \left\{\int_{\sE} g(z) \, \gamma(dz)+ \inf_{\nu \in \P(\sE)}\left(J_k(\nu)-\log \int_{\sE} e^{g(z)} \, \Gamma_k(\nu)(dz)\right)\right\} \nonumber \\
&\geq \sup_{g \in C_b(\sE)} \left\{\int_{\sE} g(z) \, \gamma(dz)+ \inf_{\nu \in \P(\sE)}\left(V_k(\nu)-\log \int_{\sE} e^{g(z)} \, \Gamma_k(\nu)(dz)\right)\right\} \nonumber\\
&= V_{k+1}(\nu),\nonumber 
\end{align}
where the last inequality follows from induction hypothesis.
Hence the claim also holds for $k+1$. This completes the proof.
\end{proof}

Proposition~\ref{proposition1}  and Proposition~\ref{proposition2} will not be used to prove the main result of the paper. However, these results establish the distinction between interacting particle systems that arise from the state-action dynamics of mean-field games and interacting particle systems that were studied in prior  literature. In particular, they show that decline of large deviations probabilities of empirical measures from limiting distributions is faster in our setting.


\section{Proof of Large Deviations Principle}\label{sec3}

To prove LDP in Theorem~\ref{theorem2}, we first use the elegant method of Delarue, Lacker, and Ramanan \cite[Section 6.3.1]{DeLaRa20}, which was suggested to prove the LDP for the continuous-time mean-field games under the common noise. Here, the idea is to transfer the LDP, satisfied by the empirical measure of initial states and noise variables due to the Sanov's theorem, to empirical measures of states.

To this end, let us define 
$
w_i^N \eqdef \{w_i^N(t)\}_{t\geq0} \,\, \text{and} \,\, v_i^N \eqdef \{v_i^N(t)\}_{t\geq0}. 
$
Note that random elements $\{(x_i^N(0),w_i^N,v_i^N)\}_{i=1}^N$ are i.i.d. with a common distribution
$
\mu(0) \bigotimes_{t=0}^{\infty} \xi_t \bigotimes_{t=0}^{\infty} \alpha_t \eqdef \mu(0) \otimes \Theta_w \otimes \Theta_v. 
$
Let us define the empirical measure of the above process:
$$
Q^N \eqdef \frac{1}{N} \sum_{i=1}^N \delta_{(x_i^N(0),w_i^N,v_i^N)}. 
$$
Let us endow the set of probability measures $\P(\sX \times \sW^{\infty} \times \sV^{\infty})$ with setwise convergence topology. Then, by Sanov's Theorem in setwise convergence topology \cite[Theorem 1.1]{Dea94}, $\{Q^N\}_{N\geq1}$ satisfies LDP on $\P(\sX \times \sW^{\infty} \times \sV^{\infty})$ with the rate function $R(\,\cdot\,|\mu(0)\times\Theta_w\times\Theta_v)$. Our first step is to transfer this LDP to empirical measures $\{(b^N(t))_{t\geq0}\}_{N\geq1}$ via the contraction principle. To this end, let us define the function 
$$
\Phi: \P(\sX\times\sW^{\infty}\times\sV^{\infty}) \rightarrow \P(\sX\times\sA)^{\infty}
$$  
as follows. Given any $Q \in \P(\sX\times\sW^{\infty}\times\sV^{\infty})$, let 
\begin{align}
\left(x(0),w(0),\ldots,v(0),\ldots\right) \sim Q \label{eq4}
\end{align}
and define recursively the following random elements:
\begin{align}
x(t+1) &= F_t\left(x(t),a(t),\L\{x(t)\},w(t)\right) \label{eq5}\\
a(t) &= G_t(x(t),v(t)) \label{eq6}
\end{align}
for $t\geq0$. Then, we define $\Phi(Q) \eqdef \left(\L\{x(t),a(t)\}\right)_{t\geq0}$. The following lemma states that the image of $Q^N$ under $\Phi$ is $(b^N(t))_{t\geq0} \eqdef b^N$.

\begin{lemma}\label{lemma1}
For each $N\geq1$, $\Phi(Q^N) = b^N$.
\end{lemma}

\begin{proof}
Any realization of 
$
Q^N \eqdef \frac{1}{N} \sum_{i=1}^N \delta_{(x_i^N(0),w_i^N,v_i^N)} 
$
must be of the form $(x_i^N(0),w_i^N,v_i^N)$ for some $i=1,\ldots,N$ with equal probability $\frac{1}{N}$. We claim that for each $t\geq0$,
$$
\L\{x(t),a(t),w(t),v(t)\} = \frac{1}{N} \sum_{i=1}^N \delta_{(x_i^N(t),a_i^N(t),w_i^N(t),v_i^N(t))}
$$
under $Q^N$. For $t=0$, with probability $\frac{1}{N}$, 
$$
x(0) = x_i^N(0), \,\, a(0) = G_0(x_i^N(0),v_i^N(0)) \eqdef a_i^N(0), \,\, w(0) = w_i^N(0), \,\, v(0) = v_i^N(0). 
$$
Hence the claim is true for $t=0$. Suppose that the claim is true for all $t\leq k$ and consider $k+1$. In this case, with probability $\frac{1}{N}$, 
\begin{align}
x(k+1) &= F_k(x_i^N(k),a_i^N(k),s^N(k),w_i^N(k)) \eqdef x_i^N(k+1), \nonumber \\
a(k+1) &= G_{k+1}(x(k+1),v(k+1)), \nonumber \\
w(k+1) &= w_i^N(k+1), \,\, v(k+1) = v_i^N(k+1),\nonumber
\end{align}
and so, $a(k+1) = G_{k+1}(x_i^N(k+1),v_i^N(k+1)) \eqdef a_i^N(k+1)$. Therefore, the claim is also true for $k+1$. By the induction hypothesis, the claim is true for all $t\geq0$. 
\end{proof}

To be able to use the contraction principle, we need to establish that $\Phi$ is continuous. Before we prove the continuity of $\Phi$, let us introduce the following notation. For each $t\geq0$, let 
$F^{t+1}: \sX \times \prod_{k=0}^t (\sW\times\sV\times\P(\sX)) \rightarrow \sX$ and 
$G^{t+1}: \sX \times \prod_{k=0}^t (\sW\times\sV\times\P(\sX)) \times \sV \rightarrow \sA$
be defined recursively as follows:
\begin{align}
F^0(x(0)) \eqdef x(0), \,\,
G^0(x(0),v(0)) &\eqdef G_0(x(0),v(0)) \nonumber \\
F^{t+1}\left(x(0),\left[w(k),v(k),\L\{x(k)\}\right]_{k=0}^t\right) &\eqdef F_t\left(F^t,G^t,\L\{x(t)\},w(t)\right) \nonumber \\
G^{t+1}\left(x(0),\left[w(k),v(k),\L\{x(k)\}\right]_{k=0}^t,v(t+1)\right) &\eqdef G_{t+1}\left(F^{t+1},v(t+1)\right). \nonumber
\end{align}
By Assumption~\ref{assumption1}, for each $t\geq0$, the mapping 
$$F^t\left(\,\cdot\,,\left[w(k),v(k)\right]_{k=0}^t\right)$$ is continuous in 
$$\left(x(0),\left[\L\{x(k)\}\right]_{k=0}^t\right)$$ 
for any $\left[w(k),v(k)\right]_{k=0}^t$; that is it is continuous for all its arguments except the noise variables, and same is true for $G^t$; that is, 
$$G^t\left(\,\cdot\,,\left[w(k),v(k)\right]_{k=0}^t,v(t+1)\right)$$ 
is continuous in 
$$\left(x(0),\left[\L\{x(k)\}\right]_{k=0}^t\right)$$ 
for any $(\left[w(k),v(k)\right]_{k=0}^t,v(t+1))$. Using above definitions, we also introduce another notation. Let 
$
f: \sX^{\infty} \times \sA^{\infty} \times \sW^{\infty} \times \sV^{\infty} \rightarrow \R
$ 
be a continuous and bounded function. Then, we introduce recursively the following functions:
\small
\begin{align}
&h_f^0\left(x(0),\ldots,a(1),\ldots,w(0),\ldots,v(0),\ldots\right) \nonumber \\
&\eqdef f\left(F^0,x(1),\ldots,G^0,a(1),\ldots,w(0),\ldots,v(0),\ldots\right) \nonumber \\
&h_f^1\left(x(0),x(2),\ldots,a(2),\ldots,w(0),\ldots,v(0),\ldots,\L\{x(0)\}\right) \nonumber \\
&\eqdef h_f^0\left(x(0),F^1,x(2),\ldots,G^1,a(2),\ldots,w(0),\ldots,v(0),\ldots\right) \nonumber \\
&\phantom{ii}\vdots \nonumber \\
&h_f^{t+1}\left(x(0),x(t+2),\ldots,a(t+2),\ldots,w(0),\ldots,v(0),\ldots,\L\{x(0)\},\ldots\,\L\{x(t)\}\right) \nonumber \\
&\eqdef h_f^t\big(x(0),F^{t+1},x(t+2),\ldots,G^{t+1},a(t+2),\ldots,w(0),\ldots,v(0),\ldots,\nonumber \\
&\phantom{xxxxxxxxxxxxxxxxxxxxxxxxxxxxxxxxxxxxxx}\L\{x(0)\},\ldots\,\L\{x(t-1)\}\big) \nonumber\\
&\phantom{ii}\vdots \nonumber
\end{align}
\normalsize
By Assumption~\ref{assumption1} and continuity of $f$, bounded functions $\{h_f^t\}_{t\geq0}$ are also continuous for all its arguments except the noise variables. It is now time to prove that $\Phi$ is continuous. 
 
\begin{proposition}\label{proposition3}
Under Assumption~\ref{assumption1}, the mapping $\Phi: \P(\sX\times\sW^{\infty}\times\sV^{\infty}) \rightarrow \P(\sX\times\sA)^{\infty}$ is continuous, where $\P(\sX\times\sW^{\infty}\times\sV^{\infty})$ is endowed with setwise convergence topology and $\P(\sX\times\sA)^{\infty}$ has product topology with weak convergence topology on $\P(\sX\times\sA)$. 
\end{proposition}

\begin{proof}
Since the setwise topology is non-metrizable for non-finite sets, we work with nets instead of sequences to prove continuity. To this end, let $\{Q^{\alpha}\}$ be a net in $\P(\sX\times\sW^{\infty}\times\sV^{\infty})$ that converges to $Q$ setwise. Define $\Phi(Q^{\alpha}) \eqdef (b^{\alpha}(t))_{t\geq0}$ and $\Phi(Q) \eqdef (b(t))_{t\geq0}$. We need to prove that $b^{\alpha}(t) \rightarrow b(t)$ weakly for each $t\geq0$ since $\P(\sX\times\sA)$ is endowed with weak convergence topology. 

Let $\{x^{\alpha}(t),a^{\alpha}(t),w^{\alpha}(t),v^{\alpha}(t)\}_{t\geq0}$ be random elements defined by equations (\ref{eq4})-(\ref{eq6}) under $Q^{\alpha}$, and let $\{x(t),a(t),w(t),v(t)\}_{t\geq0}$ be random elements defined by equations (\ref{eq4})-(\ref{eq6}) under $Q$. We claim that for each $k\geq0$, we have
$$
\L\left(\{x^{\alpha}(t),a^{\alpha}(t),w^{\alpha}(t),v^{\alpha}(t)\}_{t=0}^k\right) \rightarrow \L\left(\{x(t),a(t),w(t),v(t)\}_{t=0}^k\right)
$$
weakly. For $k=0$, let $f \in C_b(\sX\times\sA\times\sW\times\sV)$. Then we have
\small
\begin{align}
&\int f \, d\L\left(x^{\alpha}(0),a^{\alpha}(0),w^{\alpha}(0),v^{\alpha}(0)\right) = E\left[f\left(x^{\alpha}(0),G^0(x^{\alpha}(0),v^{\alpha}(0)),w^{\alpha}(0),v^{\alpha}(0)\right)\right] \nonumber \\
&= E\left[h_f^0\left(x^{\alpha}(0),w^{\alpha}(0),v^{\alpha}(0)\right)\right] 
\longrightarrow E\left[h_f^0\left(x(0),w(0),v(0)\right)\right] \nonumber \\
&= \int f \, d\L\left(x(0),a(0),w(0),v(0)\right) \nonumber
\end{align}
\normalsize
since $h_f^0$ is measurable and bounded, and $Q^{\alpha} \rightarrow Q$ setwise. Hence the claim holds for $k=0$. Suppose that the claim holds until $k$ and consider $k+1$. Let $f \in C_b(\sX^{k+1}\times\sA^{k+1}\times\sW^{k+1}\times\sV^{k+1})$. Since $Q^{\alpha}\rightarrow Q$ setwise and $$\delta_{\L\{x^{\alpha}(0)\}} \times \ldots \times \delta_{\L\{x^{\alpha}(k)\}} \rightarrow \delta_{\L\{x(0)\}} \times \ldots \times \delta_{\L\{x(k)\}}$$ 
weakly as $\L\{x^{\alpha}(t)\} \rightarrow \L\{x(t)\}$ weakly for each $0\leq t \leq k$ by the induction hypothesis, by \cite[Theorem 3.7(ii)]{Sch75}, we have  
\begin{align}
\left(\delta_{\L\{x^{\alpha}(0)\}} \times \ldots \times \delta_{\L\{x^{\alpha}(k)\}}\right)\times Q^{\alpha} \rightarrow \left(\delta_{\L\{x(0)\}} \times \ldots \times \delta_{\L\{x(k)\}}\right) \times Q \label{aa}
\end{align}
in $ws$-topology on $\P\left(\left[\P(\sX)\times\ldots\times\P(\sX)\right]\times\left[\sX\times\sW^{\infty}\times\sV^{\infty}\right]\right)$. 
Then, we have 
\footnotesize
\begin{align}
&\int f \, d\L\left(\{x^{\alpha}(t),a^{\alpha}(t),w^{\alpha}(t),v^{\alpha}(t)\}_{t=0}^{k+1}\right) \nonumber \\
&= E\left[h_f^{k+1}\left(x^{\alpha}(0),w^{\alpha}(0),\ldots,w^{\alpha}(k+1),v^{\alpha}(0),\ldots,v^{\alpha}(k+1),\L\{x^{\alpha}(0)\},\ldots,\L\{x^{\alpha}(k)\}\right)\right]
\nonumber \\
&=\int h_f^{k+1}\left(x_0,w_0,\ldots,w_{k+1},v_0,\ldots,v_{k+1},\mu_0,\ldots,\mu_k\right) \nonumber \\ &\phantom{xxxxxxxxx}dQ^{\alpha}\left(x_0,w_0,\ldots,w_{k+1},v_0,\ldots,v_{k+1}\right) \times \delta_{\L\{x^{\alpha}(0)\}}(d\mu_0) \times \ldots \times \delta_{\L\{x^{\alpha}(k)\}}(d\mu_k) 
\nonumber \\
&\overset{(a)}{\longrightarrow}
\int h_f^{k+1}\left(x_0,w_0,\ldots,w_{k+1},v_0,\ldots,v_{k+1},\mu_0,\ldots,\mu_k\right) \nonumber \\ &\phantom{xxxxxxxxx}dQ\left(x_0,w_0,\ldots,w_{k+1},v_0,\ldots,v_{k+1}\right) \times \delta_{\L\{x(0)\}}(d\mu_0) \times \ldots \times \delta_{\L\{x(k)\}}(d\mu_k) 
\nonumber \\
&=E\left[h_f^{k+1}\left(x(0),w(0),\ldots,w(k+1),v(0),\ldots,v(k+1),\L\{x(0)\},\ldots,\L\{x(k)\}\right)\right] \nonumber \\
&= \int f \, d\L\left(\{x(t),a(t),w(t),v(t)\}_{t=0}^{k+1}\right), \nonumber
\end{align}  
\normalsize
where (a) follows from (\ref{aa}) and the fact that the bounded function $h_f^{k+1}$ is continuous for all its arguments except noise variables. Therefore, by mathematical induction, this establishes the claim for each $k \geq 0$. Note that the claim implies that 
$$
b^{\alpha}(t) \eqdef \L\{x^{\alpha}(t),a^{\alpha}(t)\} \rightarrow \L\{x(t),a(t)\} \eqdef b(t)
$$
weakly for each $t\geq0$. This completes the proof. 
\end{proof}

Now using the contraction principle (i.e., Theorem~\ref{contraction}), we obtain the following result.

\begin{theorem}\label{theorem3}
The empirical measures $\left\{(b^N(t))_{t\geq0}\right\}_{N\geq1}$ satisfy the LDP on $\P(\sX\times\sA)^{\infty}$ with the rate function 
$$
J(\bnu) = \inf \left\{R(\Theta \,|\, \mu(0)\times\Theta_w\times\Theta_v): \Phi(\Theta) = \bnu \right\}.
$$
\end{theorem}


In the above characterization of the rate function $J$, we need the exact knowledge of the noise distributions $\Theta_w$ and $\Theta_v$, which may not be available. Therefore, we should replace this rate function with an equivalent one that can be expressed in terms of the components of the game dynamics. 
Note that, for any $\Theta$, by the chain rule \cite[Theorem C.3.1]{DuEl97} for relative entropy, we have 
\begin{align}
&R(\Theta\,|\,\mu(0)\times\Theta_w\times\Theta_v) = 
R(\Theta_0\,|\,\mu(0)) + \int R(\Theta^{x(0)}\,|\,\Theta_w\times\Theta_v) \, \Theta_0(dx(0)) \nonumber \\
&= R(\Theta_0\,|\,\mu(0)) + R(\Theta_0\,|\,\Theta_0) +\int R(\Theta^{x(0)}\,|\,\Theta_w\times\Theta_v) \, \Theta_0(dx(0))  \nonumber \\
&=R(\Theta_0\,|\,\mu(0)) + R(\Theta\,|\,\Theta_0\times\Theta_w\times\Theta_v). \nonumber
\end{align} 
Hence, we can alternatively write the rate function $J$ as follows:
\begin{align} \label{eq7}
J(\bnu) = R(\nu_{0,\sX}\,|\,\mu(0)) + \inf \left\{R(\Theta \,|\, \nu_{0,\sX}\times\Theta_w\times\Theta_v): \Phi(\Theta) = \bnu \right\}.
\end{align}
We now prove that the function $\inf \left\{R(\Theta \,|\, \nu_{0,\sX}\times\Theta_w\times\Theta_v): \Phi(\Theta) = \bnu \right\}$ is the rate function at $\bnu$ of the following simplified particle model. 

\subsection{Simplified Particle Model}\label{sub1sec3}

Given any $\bnu \in \P(\sX\times\sA)^{\infty}$, consider the following simplified particle model. For each $i=1,\ldots,N$, Agent~$i$ has the following state-action dynamics:
\begin{align}
\tilde{x}_i^N(0) \sim \nu_{0,\sX}, \,\, \tilde{x}_i^N(t+1) &= F_t\left(\tilde{x}_i^N(t),\tilde{a}_i^N(t),\nu_{t,\sX},\tilde{w}_i^N(t)\right) \nonumber \\
\tilde{a}_i^N(t) &= G_t\left(\tilde{x}_i^N(t),\tilde{v}_i^N(t)\right), \,\, t\geq0. \nonumber 
\end{align}
Similar to the original model, we suppose that $\{(\tilde{x}_i^N(0),\tilde{w}_i^N,\tilde{v}_i^N)\}_{i=1}^N$ are i.i.d. with the common distribution $\nu_{0,\sX} \otimes \Theta_w \otimes \Theta_v$, where 
$\tilde{w}_i^N \eqdef (\tilde{w}_i^N(t))_{t\geq0}$ and $\tilde{v}_i^N \eqdef (\tilde{v}_i^N(t))_{t\geq0}$. Let us define the empirical measure
$
\tilde{Q}^N \eqdef \frac{1}{N} \sum_{i=1}^N \delta_{(\tilde{x}_i^N(0),\tilde{w}_i^N,\tilde{v}_i^N)}. 
$
We again endow the set of probability measures $\P(\sX \times \sW^{\infty} \times \sV^{\infty})$ with setwise convergence topology. Then, by Sanov's Theorem in setwise convergence topology \cite[Theorem 1.1]{Dea94}, $\{\tilde{Q}^N\}_{N\geq1}$ satisfies the LDP on $\P(\sX \times \sW^{\infty} \times \sV^{\infty})$ with the rate function $R(\,\cdot\,|\nu_{0,\sX}\times\Theta_w\times\Theta_v)$.

In the simplified model, instead of the empirical measure $e^N(t)$ of states, we have the fixed measure $\nu_{t,\sX}$ in the state dynamics. Therefore, in this model, agents are decoupled from each other. For each $t\geq0$, let us define the state-action empirical measure 
$$
\tilde{b}^N(t) \eqdef \frac{1}{N} \sum_{i=1}^N \delta_{(\tilde{x}_i^N(t),\tilde{a}_i^N(t))}. 
$$ 
We let $\tilde{b}^N \eqdef (\tilde{b}^N(t))_{t\geq0}$. Similar to the $\Phi$, we now define 
$$\Phi^{\bnu}:\P(\sX\times\sW^{\infty}\times\sV^{\infty}) \rightarrow \P(\sX\times\sA)^{\infty}$$ 
as follows. Given any $Q \in \P(\sX\times\sW^{\infty}\times\sV^{\infty})$, let 
$$
\left(\tilde{x}(0),\tilde{w}(0),\ldots,\tilde{v}(0),\ldots\right) \sim Q
$$ 
and define recursively the following random elements:
\begin{align}
\tilde{x}(t+1) = F_t(\tilde{x}(t),\tilde{a}(t),\nu_{t,\sX},\tilde{w}(t)), \,\,\,\,
\tilde{a}(t) = G_t(\tilde{x}(t),\tilde{v}(t)), \,\, t\geq 0.\nonumber 
\end{align}
Then, we let $\Phi^{\bnu}(Q) \eqdef \left\{\L(\tilde{x}(t),\tilde{a}(t))\right\}_{t\geq0}$. The following results can be proved using the same ideas in the proofs of Lemma~\ref{lemma1} and Proposition~\ref{proposition3}, and so, we omit the details.

\begin{lemma}\label{lemma2}
For each $N\geq1$, we have $\Phi^{\bnu}(\tilde{Q}^N) = \tilde{b}^N$. 
\end{lemma}

\begin{proposition}\label{proposition4}
Under Assumption~\ref{assumption1}, the mapping $\Phi^{\bnu}: \P(\sX\times\sW^{\infty}\times\sV^{\infty}) \rightarrow \P(\sX\times\sA)^{\infty}$ is continuous, where $\P(\sX\times\sW^{\infty}\times\sV^{\infty})$ is endowed with setwise convergence topology and $\P(\sX\times\sA)^{\infty}$ has product topology with weak convergence topology on $\P(\sX\times\sA)$. 
\end{proposition}

Using LDP of the empirical process $\{\tilde{Q}^N\}_{N\geq1}$, Lemma~\ref{lemma2}, Proposition~\ref{proposition4}, and the contraction principle (i.e., Theorem~\ref{contraction}), we can conclude that the empirical measures $\{(\tilde{b}^N(t))_{t\geq0}\}_{N\geq1}$ satisfy the LDP on $\P(\sX\times\sA)^{\infty}$ with the rate function
$$
\tilde{J}({\boldsymbol \gamma}) = \inf \left\{R(\Theta \,|\, \nu_{0,\sX}\times\Theta_w\times\Theta_v): \Phi^{\bnu}(\Theta) = {\boldsymbol \gamma} \right\}.
$$
The next result will be used to connect the rate functions $J$ and $\tilde{J}$. 

\begin{lemma}\label{lemma3} 
For any $\bnu \in \P(\sX\times\sA)^{\infty}$, we have
$$
\left\{\Theta: \Phi^{\bnu}(\Theta) = \bnu \right\} = \left\{\Theta: \Phi(\Theta) = \bnu \right\}. 
$$
\end{lemma}

\begin{proof}
Let $\Theta \in \left\{\Theta: \Phi(\Theta) = \bnu \right\}$. Given $(x(0),w(0),\ldots,v(0),\ldots) \sim \Theta$, we define recursively the following random elements
\begin{align}
x(t+1) &= F_t(x(t),a(t),\L\{x(t)\},w(t)) \label{eq8} \\
a(t) &= G_t(x(t),v(t)), \,\, t\geq0. \nonumber 
\end{align}
Then, by definition $\Phi(\Theta) \eqdef \{\L(x(t),a(t))\}_{t\geq0}$. Since $\Phi(\Theta) = \bnu$, we have $\L\{x(t)\} = \nu_{t,\sX}$ for each $t\geq0$. Therefore, we can equivalently write (\ref{eq8}) as follows:
$$
x(t+1) = F_t(x(t),a(t),\nu_{t,\sX},w(t)). 
$$
This implies that $\Phi^{\bnu}(\Theta) = \bnu$ if we recall the definition of $\Phi^{\bnu}$. Hence 
$$\left\{\Theta: \Phi^{\bnu}(\Theta) = \bnu \right\} \supset \left\{\Theta: \Phi(\Theta) = \bnu \right\}.$$ The reverse inclusion can be proved similarly.
\end{proof}
In view of equation (\ref{eq7}), Lemma~\ref{lemma3} implies that
$$
J(\bnu) = R(\nu_{0,\sX}\,|\,\mu(0)) + \tilde{J}(\bnu).
$$
Note that if a process satisfies LDP with a rate function, then this rate function must be unique \cite[Theorem 1.3.1]{DuEl97}. Using this fact, we obtain another characterization of $\tilde{J}$ in terms of the components of the game dynamics. Therefore, this will lead to the characterization of $J$ in the statement of Theorem~\ref{theorem2}.

For each $i=1,\ldots,N$, let us define $\tilde{x}_i^N \eqdef (\tilde{x}_i^N(t))_{t\geq0}$ and $\tilde{a}_i^N \eqdef (\tilde{a}_i^N(t))_{t\geq0}$. Then, $\{\tilde{x}_i^N,\tilde{a}_i^N\}_{i=1}^N$ are i.i.d. with the common distribution 
\begin{align}
&\Lambda^{\bnu}(dx(0),da(0),\ldots) \eqdef \nu_{0,\sX}(dx(0)) \otimes \pi_0^{\bm}(da(0)|x(0)) \nonumber \\
&\phantom{xxxxxxxxxxxx}\bigotimes_{t=0}^{\infty} \bigg(p_t^{\nu_{t,\sX}}(dx(t+1)|x(t),a(t)) \otimes \pi_{t+1}^{\bm}(da(t+1)|x(t+1))\bigg). \nonumber 
\end{align}
For each $N\geq1$, define the following empirical measure 
$$
\tilde{e}^N \eqdef \frac{1}{N} \sum_{i=1}^N \delta_{(\tilde{x}_i^N,\tilde{a}_i^N)}. 
$$
By Sanov's Theorem in weak convergence topology \cite[Theorem 2.2.1]{DuEl97} (in this argument, we do not need Sanov's Theorem in setwise convergence topology), empirical measures $\{\tilde{e}^N\}_{N\geq1}$ satisfy the LDP on $\P(\sX^{\infty}\times\sA^{\infty})$, where $\P(\sX^{\infty}\times\sA^{\infty})$ is endowed with weak convergence topology, with the rate function $R(\,\cdot\,|\Lambda^{\bnu})$. Note that for each $k\geq0$, the $k^{th}$-marginal of $\tilde{e}^N$ on $\sX\times\sA$ is $\tilde{b}^N(k)$. Since taking marginal of a measure is continuous with respect to the weak convergence topology, by the contraction principle (i.e., Theorem~\ref{contraction}), we can conclude that $\{(\tilde{b}^N(t))_{t\geq0}\}_{N\geq1}$ satisfies the LDP on $\P(\sX\times\sA)^{\infty}$ with the rate function
$$
\tilde{V}({\boldsymbol \gamma}) = \inf \left\{R(\Lambda|\Lambda^{\bnu}): \Lambda_t = \gamma_t \,\, \forall t\geq0\right\}.
$$
Since the rate function must be unique, we must have $\tilde{V} = \tilde{J}$. Recall that $ 
J(\bnu) = R(\nu_{0,\sX}\,|\,\mu(0)) + \tilde{J}(\bnu)$. Hence, $J$ can equivalently be written as 
\begin{align}
J(\bnu) &= R(\nu_{0,\sX}\,|\,\mu(0)) + \inf \left\{R(\Lambda|\Lambda^{\bnu}): \Lambda_t = \nu_t \,\, \forall t\geq0\right\} \nonumber \\
&= \inf \left\{R(\Lambda|\Lambda^{\bnu}_{\nu(0)}): \Lambda_t = \nu_t \,\, \forall t\geq0\right\} \,\, \text{(as $\Lambda_{0,\sX} = \Lambda_{0,\sX}^{\bnu} = \nu_{0,\sX}$)} \nonumber 
\end{align}
where
$
\Lambda^{\bnu}_{\nu(0)} \eqdef \nu(0) \bigotimes_{t=0}^{\infty} (p_t^{\nu_{t,\sX}} \otimes \pi_{t+1}^{\bm}) 
$
and $\nu(0) \eqdef \mu(0) \otimes \pi_0^{\bm}$. 
This completes the proof of Theorem~\ref{theorem2}.

\begin{remark}\label{remark1}
Instead of an LDP for the marginal empirical measures, we can also establish an LDP for empirical measures of the paths using the same method. To this end, for each $N\geq1$, define the empirical measure of the paths as follows:
$$
d^N \eqdef \frac{1}{N} \sum_{i=1}^N \delta_{(x_i^N,a_i^N)},
$$
where $x_i^N \eqdef (x_i^N(t))_{t\geq0}$ and $a_i^N \eqdef (a_i^N(t))_{t\geq0}$. Using the LDP for the empirical measures of the initial states and the noise variables, we can prove that $\{d^N\}_{N\geq1}$ satisfies the LDP on $\P(\sX^{\infty}\times\sA^{\infty})$ with the rate function 
$$
L(\varphi) = \inf \{R(\Theta\,|\,\mu(0)\otimes\Theta_w\otimes\Theta_v): \Xi(\Theta)=\varphi\}, 
$$
where $\Xi:\P(\sX\times\sW^{\infty}\times\sV^{\infty})\rightarrow \P(\sX^{\infty}\times\sA^{\infty})$. Here, $\Xi$ is defined as follows: let $(x(0),w(0),\ldots,v(0),\ldots) \sim \Theta$ and define recursively random elements
\begin{align}
x(t+1) = F_t(x(t),a(t),\L\{x(t)\},w(t)), \,\,\,\,
a(t) = G_t(x(t),v(t)), \,\, t\geq0. \nonumber
\end{align}
Then, we let $\Xi(\Theta) \eqdef \L\{x(0),a(0),x(1),a(1),\ldots\}$. Now, using simplified particle model, we can indeed equivalently characterize $L$ as follows:
$$
L(\varphi) = R(\varphi|\Lambda_{\nu(0)}^{\varphi}),
$$ 
where $\Lambda_{\nu(0)}^{\varphi} \eqdef \nu(0)\bigotimes_{t=0}^{\infty}(p_t^{\varphi_{t,\sX}}\otimes\pi_{t+1}^{\bm})$ and $\varphi_t$ is the $t^{th}$ marginal of $\varphi$ on $\sX\times\sA$. The same result has been proven in \cite{MoZa03} under more restrictive conditions on the system components (see \cite[p. 59]{MoZa03}). There, authors have established LDP for $\{d^N\}_{N\geq1}$ by transferring LDP for the simplified particle model introduced in Section~\ref{sub1sec3} to the mean-field game model by generalizing Laplace-Varadhan integral lemma. In our case, we transfer  LDP for initial states and noise variables to the original mean-field game model. This enables us to establish LDP under milder conditions on the system components.

\end{remark}

\section{Conclusion}\label{conc}

In this paper, we have developed a large deviations principle for empirical measures of state-action pairs of mean-field games under the mean-field equilibrium policy. The large deviations result has been established via contraction principle by transferring LDP for empirical measures of  initial states and noise variables. One interesting future research direction is to study the concentration bound for the same empirical measures. However, in this case, we need to establish that the mean-field equilibrium policy is Lipschitz continuous, which can, in general, be proved under quite restrictive convexity assumptions on the system components.


\providecommand{\bysame}{\leavevmode\hbox to3em{\hrulefill}\thinspace}
\providecommand{\MR}{\relax\ifhmode\unskip\space\fi MR }
\providecommand{\MRhref}[2]{%
  \href{http://www.ams.org/mathscinet-getitem?mr=#1}{#2}
}
\providecommand{\href}[2]{#2}

\end{document}